\documentclass[12pt]{article}
\usepackage{amsmath, amsfonts, amssymb, amsthm, graphicx, geometry, arydshln, umoline, subfig, newtxtext, newtxmath, color, comment, soul, enumerate, bm}
\usepackage{appendix}
\usepackage{natbib}
\usepackage[colorlinks=true,citecolor=blue]{hyperref}
\usepackage{booktabs}
\usepackage[table]{xcolor}
\def\BibTeX{{\rm B\kern-.05em{\sc i\kern-.025em b}\kern-.08em
    T\kern-.1667em\lower.7ex\hbox{E}\kern-.125emX}}
\usepackage{makecell}
\usepackage{bbm}

\theoremstyle{definition}
\newtheorem{theorem}{Theorem}

\newtheorem{corollary}{Corollary}

\newtheorem{proposition}{Proposition}
\newtheorem{remark}{Remark}

\newcommand{\argmax}{\mathop{\rm arg~max}\limits}

\begin{document}

\title{Scoring Rules as Least-Squares Estimators}
\author{
Satoru Fujishige\thanks{Research Institute for Mathematical Sciences, Kyoto University, Kyoto 606-8502, Japan. Email: fujishig@kurims.kyoto-u.ac.jp}
\qquad
Satoshi Nakada\thanks{School of Management, Department of Business Economics, Tokyo University of Science, 1-11-2, Fujimi, Chiyoda-ku, Tokyo, 102-0071, Japan. Email: snakada@rs.tus.ac.jp}
}
\date{October 17, 2026}
\maketitle

\begin{abstract}
\cite{Kawada2018} proved that a scoring rule represented by a strictly positive score vector is equivalent to the corresponding cosine similarity rule.
The original proof relies on a direct analysis of the cosine similarity optimization problem. 
In this note, we present an alternative, simpler proof based on a basic least-squares characterization.
For any fixed preference profile and score vector, the arithmetic mean of the individual score vectors uniquely minimizes the sum of squared Euclidean distances to those vectors.
It then shows that, for any strictly positive affine representative of the given score vector, the optimal solutions of the cosine similarity problem are precisely the positive multiples of the arithmetic mean of the corresponding score vectors.
Since the two arithmetic means are positively affinely related, they induce exactly the same social ranking.
This result provides a transparent geometric interpretation of scoring rules and clarifies why the cosine similarity rule necessarily coincides with the corresponding scoring rule.
\\
\end{abstract}

\section{Introduction}\label{sec: introduction}
Scoring rules constitute one of the most fundamental classes of voting rules.\footnote{See \citet{Fishburn1973} for a general introduction to the literature.}
Each alternative receives a score according to its position in every individual ranking, and the social ranking is determined by the aggregate scores.
Among them, the Borda rule is arguably the best-known example and has been studied extensively from axiomatic and geometric viewpoints \citep{Young1974, NitzanRubinstein1981, Saari1990, Saari1995}.
More recently, \citet{Maskin2025} characterized the Borda rule using a modified version of Arrow's independence condition for a continuum of voters, and \citet{BarokasNitzan2025} extended the characterization to finite societies.

Recently, \cite{Kawada2018} introduced the cosine similarity rule and established an equivalence between the cosine similarity rule and the corresponding scoring rule when the score vector is strictly positive.
In particular, the cosine similarity rule associated with a strictly positive representation of the Borda score coincides exactly with the Borda rule.
This result reveals a close connection between scoring rules and cosine similarity.

The purpose of this note is to give a least-squares formulation of Kawada's equivalence theorem.
Our proof is based on a simple least-squares characterization: the arithmetic mean uniquely minimizes the sum of squared Euclidean distances.
We first show that the arithmetic mean of the individual score vectors induced by the original score vector is the unique least-squares estimator of these vectors.
We next show that, for any strictly positive affine representative of the given score vector, all optimizers of the corresponding cosine similarity problem are positive multiples of the arithmetic mean of the score vectors induced by that representative.
Finally, the two arithmetic means are positively affinely related and therefore induce exactly the same social ranking.
This structure generalizes the relation between the centered Borda scores and the inverse rankings used in the Borda case.

The least-squares formulation considered here is related to earlier distance- and mean-based interpretations of voting rules \citep{CookSeiford1982, Zwicker2008, Pivato2013}.
More broadly, optimization-based formulations have long played an important role in social choice.
Classic examples include the Kemeny--Young rule \citep{Kemeny1959,Young1988}, while recent work by \cite{fujishige2026new} shows that similar optimization viewpoints can also be developed for the Schulze rule \citep{Schulze1997}.
The specific point of the present note is to show that scoring rules themselves can be viewed as least-squares estimators, thereby providing a geometric explanation of Kawada's equivalence theorem.

The remainder of the paper is organized as follows.
Section \ref{sec: leastsquare} presents the least-squares characterization.
Section \ref{sec: scoring} reviews scoring rules and the cosine similarity rule, characterizes the solutions of the cosine similarity problem, and establishes the main equivalence result.
Section \ref{sec: conclude} concludes.

\section{A Basic Least-Squares Characterization}\label{sec: leastsquare}

Throughout the paper, vectors in $\mathbb{R}^n$ are regarded as column vectors, and $x^\top$ denotes the transpose of a vector $x$.
When $X=\{1,\ldots,k\}$ for a positive integer $k$, we also write $\mathbb{R}^X$ as $\mathbb{R}^k$.

Let $N=\{1,\ldots,n\}$ be the set of voters and $V=\{a_1,\ldots,a_m\}$ be the finite set of alternatives.
Suppose that we are given vectors $p_i\in\mathbb{R}^V$ for each $i\in N$.
The proof of our main theorem relies on the following well-known least-squares characterization.
Consider the optimization problem
\begin{equation}
\min_{x\in\mathbb R^{V}} \Phi(x)=\sum_{i\in N}\|x-p_i\|^2,
\label{eq:Phi}
\end{equation}
where $||\cdot||$ denotes the Euclidean norm, i.e., $||z||^2=\sum_{j\in V}z_j^2$ for any $z\in\mathbb{R}^V$.

Put $\bar{p}=\sum_{i\in N}p_i.$
Note that $\Phi(x)$ can be written as
\begin{align*}
\Phi(x)
&=
\sum_{i\in N}\left( \|x\|^2-2p_i^\top x+\|p_i\|^2 \right)\\
&=n\left( \|x\|^2-2\bar p^\top x \right)+\sum_{i\in N}\|p_i\|^2\\
&=n\|x-\bar p\|^2-n\|\bar p\|^2+\sum_{i\in N}\|p_i\|^2.
\end{align*}
Hence $\Phi(x)$ is minimized if and only if $\|x-\bar p\|^2$ is minimized, which is uniquely attained at $x=\bar p=(1/n)\sum_{i\in N}p_i$ since the squared Euclidean norm is strictly convex.
This elementary fact that the arithmetic mean uniquely minimizes the total squared Euclidean distance is summarized in the following proposition.

\begin{proposition}\label{prop: leastsquare}
The arithmetic mean $\bar p=(1/n)\sum_{i\in N}p_i$ is the unique solution of the optimization problem \eqref{eq:Phi}.
\end{proposition}

The expansion above also gives the following consequence, which will be used for the cosine similarity problem.

\begin{corollary}\label{cor: sphere}
Let $\rho=\|\bar p\|$.
Then $\bar p$ is the unique solution of the optimization problem
\begin{equation}
\max_{x\in S_\rho}
\left(\sum_{i\in N}p_i\right)^\top x,
\label{eq: sphere_linear}
\end{equation}
where $S_\rho=\left\{x\in\mathbb{R}^V\mid \|x\|=\rho\right\}$.
\end{corollary}

\begin{proof}
It follows from Proposition \ref{prop: leastsquare} that $\bar{p}$ is the unique minimizer of $\Phi(x)$ in $\mathbb{R}^V$. 
Hence, even if we restrict  $||x||$ to be equal to $\bar{p}$, i.e., $x\in S_\rho$, $\bar{p}$ is the unique  minimizer of $\Phi(x)$ in  $S_\rho\subset \mathbb{R}^V$ as well. 
Because of the second equality in the expansion of $\Phi(x)$ given above, minimizing  $\Phi(x)$ over $S_\rho$ is equivalent to maximizing  $\left(\sum_{i\in N}p_i\right)^\top x$ over $S_\rho$. 
The present corollary  thus follows.
\end{proof}

\section{Scoring Rules and Cosine Similarity}\label{sec: scoring}
\subsection{Scoring rules}
For every $z\in\mathbb{R}^V$, define the ranking function
$r_z:V\rightarrow\{1,\ldots,m\}$ by
\[
r_z(a)=1+
\left|
\left\{
b\in V
\mid
z(b)>z(a)
\right\}
\right|
\qquad
\forall a\in V.
\]
Then
\[
r_z(a)\le r_z(b)
\quad\Longleftrightarrow\quad
z(a)\ge z(b)
\qquad
\forall a,b\in V.
\]

Suppose that each voter has a strict preference ordering (i.e., a complete, transitive, and asymmetric binary relation) over $V$.
Since every strict preference ordering uniquely determines its strict ranking function $r: V\rightarrow\{1,\ldots,m\}$, that is, $r(a) \neq r(b)$  for any $a, b \in V$ with $a \neq b$ (and vice versa), we identify a strict preference ordering with its  associated strict ranking function.
Let $\mathcal{R}$ denote the set of all ranking functions and its subset  $\mathcal{P} \subsetneq \mathcal{R}$ denote the set of all strict ranking functions.
Here, $r(a)=1$ means that alternative $a$ is ranked first (most preferred) and $r(a)=m$ means that it is ranked last (least preferred).
The social welfare function $F: \mathcal{P}^n \rightarrow \mathcal{R}$ determines a social ranking depending on the ranking profiles $(r_i)_{i \in N}$ of voters.

Fix a score vector $\theta=(\theta_1,\ldots,\theta_m)\in\mathbb R^m$ satisfying $\theta_1\ge\theta_2\ge\cdots\ge\theta_m$.
For each voter $i\in N$ with ranking $r_i\in \mathcal{P}$, the corresponding score function $s^{\theta}_i:V \rightarrow \mathbb{R}$ is defined by
\[
s_i^\theta(a)=\theta_{r_i(a)} \qquad \forall a\in V.
\]
The aggregate score vector is
\[
S^\theta=\sum_{i\in N}s_i^\theta\in\mathbb{R}^V,
\]
so that the aggregate score of an alternative $a$ is $S^\theta(a)=\sum_{i\in N}s_i^\theta(a)=\sum_{i\in N}\theta_{r_i(a)}$.
The scoring rule associated with $\theta$ is the social welfare function that assigns to every ranking profile the social ranking $r^\theta=r_{S^\theta}=r_{\bar s^\theta}$.
Equivalently,
\begin{equation}
r^\theta(a)\le r^\theta(b)
\quad\Longleftrightarrow\quad
S^\theta(a)\ge S^\theta(b)
\qquad
\forall a,b\in V.
\label{eq:scoringrule}
\end{equation}
The most prominent example is the \emph{Borda rule}, which corresponds to the score vector $\theta_k=m+1-k$ for any $k=1,\ldots,m$.

Let
\[
\bar{s}^{\theta}=\frac{1}{n}\sum_{i\in N}s_i^\theta=\frac{1}{n}S^\theta
\]
be the arithmetic mean of the score vectors.
The following result applies Proposition \ref{prop: leastsquare} directly to a scoring rule.

\begin{theorem}\label{thm: scoring_leastsquare}
The arithmetic mean $\bar{s}^{\theta}$ is the unique solution of the optimization problem
\begin{equation}
\min_{x\in\mathbb{R}^V}
\sum_{i\in N}\|x-s_i^\theta\|^2.
\label{eq: scoring_leastsquare}
\end{equation}
Moreover, the scoring rule associated with $\theta$ is precisely the social ranking induced by the coordinates of $\bar{s}^{\theta}$.
\end{theorem}

\begin{proof}
The first statement follows from Proposition \ref{prop: leastsquare} by setting $p_i=s_i^\theta$ for every $i\in N$.
Since $S^\theta=n\bar{s}^{\theta}$, we have
\[
r^\theta(a)\le r^\theta(b)
\quad\Longleftrightarrow\quad
\bar{s}^{\theta}(a)\ge \bar{s}^{\theta}(b)
\qquad \forall a,b\in V,
\]
which proves the second statement.
\end{proof}

\subsection{Cosine Similarity Rules}

A score vector $\eta=(\eta_1,\ldots,\eta_m)$ is called a \emph{positive affine representative} of $\theta$ if
\begin{equation}
\eta=\alpha\theta+ {\bf 1}_V \in \mathbb{R}_{++}^V
\label{eq: positive_representative}
\end{equation}
for some $\alpha>0$ and $\beta\in\mathbb{R}$.
Here, $\mathbb{R}_{++}^X=\left\{x\in\mathbb{R}^X\mid x(j)>0\text{ for every }j\in X\right\}$ for any finite index set $X$, and ${\bf 1}_m$ denotes the all-one vector in $\mathbb{R}^m$.
Every score vector has a positive affine representative; for example,
\[
\eta=\theta+(1-\theta_m){\bf 1}_V \ge {\bf 1}_V.
\]

For every voter $i\in N$, \eqref{eq: positive_representative} implies
\begin{equation}
s_i^\eta
=
\alpha s_i^\theta+\beta {\bf 1}_V,
\label{eq: individual_affine}
\end{equation}
where ${\bf 1}_V$ denotes the all-one vector in $\mathbb{R}^V$.
Consequently,
\begin{equation}
\bar{s}^{\eta}
=
\alpha\bar{s}^{\theta}+\beta {\bf 1}_V.
\label{eq: average_affine}
\end{equation}
In particular, the scoring rules associated with $\theta$ and $\eta$ coincide by Theorem \ref{thm: scoring_leastsquare}.

We next recall the cosine similarity rule introduced by \cite{Kawada2018}.
For any non-zero vectors $x,y\in\mathbb R^{V}$, their cosine similarity is defined as
\[
C(x,y)=\frac{x^\top y}{\|x\|\|y\|}.
\]
Fix a positive affine representative $\eta$ of $\theta$ and consider the following optimization problem:
\begin{equation}
\max_{x\in\mathbb{R}_{++}^V}
\sum_{i\in N} C(s^{\eta}_i,x).
\label{eq:cosproblem}
\end{equation}

The use of a positive affine representative is essential for the cosine similarity formulation.

\begin{remark}\label{rem: positivity}
Consider $m=3$, $\theta=(1,0,-1)$, and the three-voter profile
\[
a\succ_1 b\succ_1 c,
\qquad
b\succ_2 c\succ_2 a,
\qquad
c\succ_3 a\succ_3 b.
\]
Then $S^\theta={\bf 0}$, and hence
\[
\sum_{i\in N}C(s_i^\theta,x)=0
\]
for every non-zero $x\in\mathbb{R}^V$.
Thus a cosine similarity rule based directly on this numerical representation is not well-defined.
On the other hand, $\eta=(3,2,1)=\theta+2{\bf 1}_{\{1,2,3\}}$ is a positive affine representative and $S^\eta=(6,6,6)$, so the corresponding cosine similarity problem induces the social ranking that regards all alternatives as indifferent, as does the scoring rule.
\end{remark}

The next theorem characterizes all optimal solutions of \eqref{eq:cosproblem}.

\begin{theorem}\label{thm: cosine_solution}
Let $\eta\in\mathbb{R}^V_{++}$ be a score vector.
Then
\begin{equation}
\argmax_{x\in\mathbb{R}_{++}^V}
\sum_{i\in N}C(s_i^\eta,x)
=
\left\{\lambda\bar{s}^{\eta}\mid \lambda>0\right\}.
\label{eq: cosine_solution}
\end{equation}
\end{theorem}

\begin{proof}
Since $\eta\in\mathbb{R}^V_{++}$, we have $\bar{s}^{\eta}\in\mathbb{R}_{++}^V$.
In particular, $\rho=\|\bar{s}^{\eta}\|>0$.
By Corollary \ref{cor: sphere}, $\bar{s}^{\eta}$ is the unique maximizer of
\[
\left(\sum_{i\in N}s_i^\eta\right)^\top x
\]
over the sphere $S_\rho=\{x\in\mathbb{R}^V\mid\|x\|=\rho\}$.

Since each $s_i^\eta$ is a permutation of $\eta$,
\[
\|s_i^\eta\|
=
\left(\sum_{k=1}^m\eta_k^2\right)^{1/2}
=:R_\eta,
\]
which is independent of the voter $i$.
Hence, for every $x\in S_\rho$,
\[
\begin{aligned}
\left(\sum_{i\in N}s_i^\eta\right)^\top x
&=\sum_{i\in N}(s_i^\eta)^\top x\\
&=\sum_{i\in N}\|s_i^\eta\|\,\|x\|\,C(s_i^\eta,x)\\
&=R_\eta\rho\sum_{i\in N}C(s_i^\eta,x).
\end{aligned}
\]
Since $R_\eta\rho$ is a positive constant, $\bar{s}^{\eta}$ is the unique maximizer of the cosine similarity objective over $S_\rho$.

Finally, the cosine similarity objective is invariant under positive rescaling of $x$.
For any $x\in\mathbb{R}_{++}^V$, the vector $(\rho/\|x\|)x$ belongs to $S_\rho\cap\mathbb{R}_{++}^V$ and has the same objective value as $x$.
It follows that every optimal solution of \eqref{eq:cosproblem} is a positive multiple of $\bar{s}^{\eta}$ and that every such positive multiple is optimal.
\end{proof}

By Theorem \ref{thm: cosine_solution}, all optimal solutions of \eqref{eq:cosproblem} induce the same coordinate-wise order.
Let $x^*$ be any optimal solution of \eqref{eq:cosproblem}.
The cosine similarity rule associated with $\eta$ ranks alternatives according to the coordinates of $x^*$; that is, $r^{\mathrm{cos},\eta}=r_{x^*}$ and 
\begin{equation}
r^{\mathrm{cos}, \eta}(a) \le r^{\mathrm{cos}, \eta}(b)
\quad\Longleftrightarrow\quad
x^*(a) \ge x^*(b)
\qquad \forall a,b\in V.
\label{eq:cosrule}
\end{equation}

\subsection{Equivalence}

We now establish the equivalence between the scoring rule associated with an arbitrary score vector and the cosine similarity rule associated with any of its positive affine representatives.

\begin{theorem}\label{thm: main}
Let $\theta$ be a score vector, and let $\eta$ be a positive affine representative of $\theta$.
Then the scoring rule associated with $\theta$, \eqref{eq:scoringrule}, coincides with the cosine similarity rule associated with $\eta$, \eqref{eq:cosrule}.
In particular, the social ranking induced by the cosine similarity problem is independent of the choice of a positive affine representative of $\theta$.
\end{theorem}

\begin{proof}
By Theorem \ref{thm: cosine_solution}, every optimal solution of \eqref{eq:cosproblem} is of the form
\[
x^*=\lambda\bar{s}^{\eta}
\]
for some $\lambda>0$.
By \eqref{eq: average_affine},
\[
\bar{s}^{\eta}=\alpha\bar{s}^{\theta}+\beta {\bf 1}_V
\]
for some $\alpha>0$ and $\beta\in\mathbb{R}$.
Therefore, for every $a,b\in V$,
\[
\begin{aligned}
x^*(a)-x^*(b)
&=\lambda\left(\bar{s}^{\eta}(a)-\bar{s}^{\eta}(b)\right)\\
&=\lambda\alpha\left(\bar{s}^{\theta}(a)-\bar{s}^{\theta}(b)\right).
\end{aligned}
\]
Since $\lambda\alpha>0$, we have
\[
x^*(a)\ge x^*(b)
\quad\Longleftrightarrow\quad
\bar{s}^{\theta}(a)\ge\bar{s}^{\theta}(b)
\quad\Longleftrightarrow\quad
S^\theta(a)\ge S^\theta(b).
\]
Hence the cosine similarity rule associated with $\eta$ coincides with the scoring rule associated with $\theta$.
The same argument applies to every positive affine representative of $\theta$, which proves the last statement.
\end{proof}

The result of \cite{Kawada2018} follows by taking the score vector itself as its positive affine representative.

\begin{corollary}[\citealt{Kawada2018}]\label{cor: Kawada}
Let $\theta\in\mathbb{R}^V_{++}$ be a score vector.
Then the scoring rule associated with $\theta$ coincides with the cosine similarity rule obtained directly from the score vectors $(s_i^\theta)_{i\in N}$.
\end{corollary}

The relation with the original Borda construction can also be seen directly.

\begin{remark}[The Borda case]\label{rem: Borda}
Consider the centered Borda score vector $\theta^B$ and the positive Borda score vector $\eta^B$ defined by
\[
\theta_k^B=m+1-2k
\qquad\text{and}\qquad
\eta_k^B=m+1-k
\qquad (k=1,\ldots,m).
\]
They satisfy
\[
\eta^B
=
\frac{1}{2}\theta^B+\frac{m+1}{2}{\bf 1}_V.
\]
Consequently, for every voter $i\in N$,
\[
2s_i^{\eta^B}
=
s_i^{\theta^B}+(m+1){\bf 1}_V,
\]
and hence
\[
2\bar{s}^{\eta^B}
=
\bar{s}^{\theta^B}+(m+1){\bf 1}_V.
\]
Here, $s_i^{\theta^B}$ is the centered Borda score vector and $s_i^{\eta^B}(a)=m+1-r_i(a)$ is the inverse ranking used in the cosine similarity formulation.
Thus Theorem \ref{thm: main} specializes to the affine relation between the averaged Borda score and the averaged inverse ranking.
\end{remark}

\section{Concluding Remarks}\label{sec: conclude}
In this note, we provide a least-squares formulation of the equivalence established by \cite{Kawada2018}.
The arithmetic mean of the original score vectors is the unique least-squares estimator and induces the corresponding scoring rule.
For any positive affine representative, the optimal solutions of the cosine similarity problem are precisely the positive multiples of the arithmetic mean of the corresponding score vectors.
The two arithmetic means are positively affinely related and therefore induce exactly the same social ranking.
This separation clarifies why strict positivity is required in the cosine similarity problem while the scoring rule itself is unchanged under positive affine transformations of its score vector.

This interpretation suggests several possible directions for future research.
For example, it would be interesting to investigate whether similar geometric arguments can be applied to weighted scoring rules, incomplete rankings, or other similarity measures.
More generally, the present approach may provide a useful framework for understanding optimization-based representations of social choice procedures.

\subsection*{Acknowledgements}

S.~Fujishige's research was supported by JSPS KAKENHI Grant Number JP22K11922 and by JST ERATO Grant Number JPMJER2301, and also by the Research Institute for Mathematical Sciences, an International Joint Usage/Research Center located in Kyoto University.
S. Nakada's research was supported by JSPS KAKENHI Grant Number 25K16606 and 25K00618.

\subsection*{Declaration}
The authors have no financial or proprietary interests in any material discussed in this article, and there is no competing interest to declare.

\bibliography{reference}

\end{document}